\documentclass[10pt]{article}
\usepackage{epsfig,amsmath,graphicx,amssymb,overpic}

\usepackage{graphicx}
\usepackage{dcolumn}
\usepackage{bm}
\usepackage{graphicx}
\usepackage{subfigure}
\usepackage{epsfig,amsmath,graphicx,amssymb,overpic,cite}

\usepackage{graphicx}
\usepackage{dcolumn}
\usepackage{bm}
\usepackage{graphicx}
\usepackage{subfigure}
\usepackage{amsthm}

\newtheorem{corollary}{Corollary}
\newtheorem{lemma}{Lemma}
\newtheorem{proposition}{Proposition}

\newtheorem{theorem}{Theorem}

\def\be{\begin{equation}}
\def\ee{\end{equation}}
\def\bee{\begin{eqnarray}}
\def\ene{\end{eqnarray}}
\def\bes{\begin{subequations}}
\def\ees{\end{subequations}}

\def\v{\vspace{0.1in}}
\def\no{{\nonumber}}
\def\PT{{\cal PT}}

\begin{document}

\baselineskip=13pt
\renewcommand {\thefootnote}{\dag}
\renewcommand {\thefootnote}{\ddag}
\renewcommand {\thefootnote}{ }

\pagestyle{plain}

\begin{center}
\baselineskip=16pt \leftline{} \vspace{-.3in} {\Large \bf A unified inverse scattering transform and soliton solutions of the nonlocal mKdV equation with non-zero boundary conditions} \\[0.2in]
\end{center}

\begin{center}
Guoqiang Zhang and  Zhenya Yan$^{*}$\footnote{$^{*}${\it Email address}: zyyan@mmrc.iss.ac.cn}\\[0.03in]
{\it \small Key Laboratory of Mathematics Mechanization, Academy of Mathematics and Systems Science, Chinese Academy of Sciences, Beijing 100190, China \\
 School of Mathematical Sciences, University of Chinese Academy of Sciences, Beijing 100049, China} \\
\end{center}

\vspace{0.1in}

{\baselineskip=14pt




{\bf Abstract.}\, In this paper, we present a systematical theory of a unified and simple inverse scattering transform for both focusing and defocusing nonlocal (reverse-space-time) modified Korteweg-de Vries (mKdV) equations with non-zero boundary conditions (NZBCs) at infinity. The suitable uniformization variable is introduced to make the direct and inverse scattering problems be established on a new complex plane instead of the Riemann surface. The direct scattering problem establishes the analyticity, symmetries, and asymptotic behaviors of Jost solutions and scattering matrix, and properties of discrete spectra. The inverse scattering problem can be solved by means of a corresponding matrix-valued Riemann-Hilbert problem. The reconstruction formula, trace formulae, and theta conditions are found. Finally, the dynamical behaviors of solitons and their interactions for four distinct cases for the reflectionless potentials for both focusing and defocusing nonlocal mKdV equations with NZBCs are analyzed in detail.

\vspace{0.1in} \noindent {\it Keywords:}  nonlocal mKdV equation;  non-zero boundary conditions; Riemann surface; inverse scattering transform; matrix Riemann-Hilbert problem; solitons, breathers



\vspace{0.1in}

\section{Introduction}

Generally speaking, it is so difficult to solve exactly nonlinear partial differential equations (PDEs) with initial-value conditions, but,
a novel inverse scattering transform (IST) was discovered by Gardner, Greene, Kruskal and Miura (GGKM) to exactly solve the initial-value
problems for the important Korteweg-de Vries (KdV) equation with a Lax pair in 1967~\cite{Gardner1967}. After that, numerous attempts were made to extend the application of this method in other integrable nonlinear PDEs admitting the  so-called Lax pairs~\cite{Lax1968}. In 1972, Zakharov and Shabat  investigated the IST of the nonlinear Schr\"odinger (NLS) equation~\cite{Shabat1972}. After that, Ablowitz, Kaup, Newell and Segur (AKNS)
presented a class of new integrable systems, called AKNS systems, and found a general framework for their ISTs~\cite{Ablowitz1973, Ablowitz1974}. Subsequently, many integrable nonlinear wave equations were shown to be solved in terms of the IST, such as the modified KdV equation \cite{Wadati1973, Wadati1982}, the sine-Gordon equation \cite{Ablowitz1973a}, the Kadomtsev-Petviashvili equation \cite{Ablowitz1983}, the Camassa-Holm equation \cite{Constantin2006}, the Benjamin-Ono Equation \cite{Fokas1983} and the Degasperis-Procesi equation \cite{Constantin2010}.

Since the parity-time ($\PT$) symmetry was introduced in the generalized Hamiltonians in 1998 by Bender {\it et al}~\cite{bender1} and in the nonlinear Sch\"odinger equation\cite{nlspt}, it has been verified to play a more and more important role in many fields (see, e.g.,~\cite{bender2, vvk}). In 2013, the $\PT$ symmetry was introduced to the first one of the well-known AKNS system to present a nonlocal NLS equation, and outlined the IST with zero-boundary conditions (ZBCs)~\cite{Ablowitz2013}. The nonlocal NLS equation was shown to be equivalent to the unconventional system of  Landau-Lifshitz equations under the sense of gauge transformation~\cite{Ga}. The nonlocal integrable systems are of important significance in the theoretical study of mathematical physics and applications in the fields of nonlinear science~\cite{yang18}. Moreover, the multi-component local and nonlocal generalized nonlinear Schr\"odinger (NLS) equations were recently introduced with the aid of two families of parameters\cite{yan15, yan16, yan18}. Some other nonlocal nonlinear wave equations were presented (see, e.g., Refs.~\cite{n1,n2,n3,n4,n5,n6,n7}). Some reverse time-space and inverse time integrable nonlocal nonlinear wave equations were found, and their ISTs with ZBCs were presented~\cite{Ablowitz2016, Ablowitz2017}. Recently, Ablowitz {\it et al.} developed the IST with non-zero boundary conditions (NZBCs) for the nonlocal NLS equation, another nonlocal (reverse time-space) NLS equation and nonlocal reverse time-space sine-Gordon/sinh-Gordon equation \cite{Ablowitz2018, Ablowitz2018b, Ablowitz2018a}. In contrast to the original method for the integrable systems with NZBCs developed by Zakharov \cite{Zakharov1973} using a two-sheeted Riemann surface, Ablowitz {\it et al.} introduced a uniformization variable \cite{Faddeev1987} to solve the inverse problem on a standard complex $z$-plane. This manner was also used to analyze the IST of the NLS equation with NZBCs by Ablowitz, {\it et al.}~\cite{Prinari2006, Ablowitz2007, Prinari2011, Demontis2013, Demontis2014, Biondini2014, Biondini2014a, Kraus2015, Prinari2015, Prinari2015a, Mee2015, Biondini2015, Biondini2016, Biondini2016a, Pichler2017, Prinari2018}. Recently, we developed the approach to present a systematical theory for the IST of both focusing and defocusing modified KdV equations with NZBCs at infinity~\cite{yanzhang18}.

Recently, an integrable real nonlocal (also called reverse-space-time) mKdV equation was introduced~\cite{Ablowitz2016, Ablowitz2017}
\begin{gather}\label{n-mKdV}
\begin{gathered}
 q_t(x, t)-6\,\sigma q(-x, -t)\,q(x, t)q_x(x,t)+q_{xxx}(x, t)=0, \quad  x, t\in\mathbb{R},
\end{gathered}
\end{gather}
where $\sigma=\pm 1$ denote the focusing and defocusing cases, $q(x,t)$ is a real function, whose general form was shown to appear in the nonlinear oceanic and atmospheric dynamical system~\cite{tang}. The Darboux transformation was used to seek for soliton solutions of the focusing Eq.~(\ref{n-mKdV}) with $\sigma=-1$~\cite{zhu17}. Moreover, the IST for the focusing ($\sigma=-1$ ) Eq.~(\ref{n-mKdV}) and ZBCs was presented~\cite{zhu17b}. Particularly, i) as $q(-x, -t)=q(x, t)$,  Eq.~(\ref{n-mKdV}) becomes the usual mKdV equation; ii)
as $q(-x, -t)=-q(x, t)$, the focusing (defocusing) nonlocal mKdV equation Eq.~(\ref{n-mKdV}) reduces to the defocusing (focusing) mKdV equation;
iii) for $q(-x, -t)=-q(-x, -t)$, the focusing (defocusing) nonlocal mKdV equation Eq.~(\ref{n-mKdV}) reduces to the defocusing (focusing) nonlocal mKdV equation.

As far as we know, {\bf the IST of the nonlocal mKdV equation (\ref{n-mKdV}) with NZBCs was not reported before. Moreover, the IST of the nonlocal mKdV equation with NZBCs is more complicated than one of the nonlocal mKdV equation with ZBCs~\cite{zhu17b}}. In the following, we would like to focus on the ISTs for both focusing and defocusing nonlocal mKdV equations (\ref{n-mKdV}) with the following NZBCs
\begin{gather}\label{n-mKdV-c}
\begin{gathered}
\lim_{x\to\pm\infty}q(x,t)=q_{\pm},\quad \left|q_{\pm}\right|=q_0>0,
\end{gathered}
\end{gather}
where $q_+=\delta q_-$ with $\delta=\pm 1$. In contrast to four cases used to deal with the nonlocal NLS equation with NZBCs~\cite{Ablowitz2018}, that is, two distinct nonlinearities and two different phase differences, we will present a uniform and simple IST to explore nonlocal integrable nonlinear systems with NZBCs, such as the nonlocal mKdV equation (\ref{n-mKdV}) with NZBCs (\ref{n-mKdV-c}).

The remaining part of this paper is organized as follows. In Sec. 2, we deduce the direct scattering problem for both focusing and defocusing nonlocal mKdV equations with NZBCs. A uniformization variable is introduced such that the direct scattering problem can be studied in a standard complex $z$-plane. Then the analytical domains of the Jost solutions and scattering data are found. Based on the reduction conditions of the Lax pair, the basic symmetries of the modified Jost solutions and scattering matrix are also established, which generate the discrete spectrum and residue conditions. Moreover, the asymptotic behaviors for the modified Jost solutions and scattering matrix are also discovered. Sec. 3 focuses on the inverse problem with NZBCs. A generalized and uniform matrix-valued Riemann-Hilbert problem (RHP) is formulated and can be solved by means of the Cauchy projectors and Plemelj's formulae.
The trace formula and theta condition are found. In Sec. 4, the reflectionless potentials for both focusing and defocusing nonlocal mKdV equations with NZBCs are obtained such that some soliton solutions and breathers, and their interactions are illustrated. Finally, we give some conclusions and discussions in Sec. 5.

\section{Direct scattering problem with NZBCs}

In the direct scattering theory with NZBCs, we will deduce the analytical properties, symmetries and asymptotic behaviors of (modified) Jost solutions and scattering coefficients, and the discrete spectrum.

\subsection{Lax pair and Riemann surface with NZBCs}

The nonlocal mKdV equation (\ref{n-mKdV}) possesses the nonlocal Lax pair~\cite{Ablowitz2016}
\begin{alignat}{2}
\label{lax-x}
\varPhi_x&=X\varPhi, &\quad X=X(x ,t; k):&=ik\sigma_3+Q,  \\ \label{lax-t}
\varPhi_t&=T\varPhi, &\quad T=T(x, t; k):&=\left[4k^2+2\,\sigma q(x, t)\,q(-x, -t)\right]X-2ik\sigma_3Q_x+[Q_x, Q]-Q_{xx},
\end{alignat}
where $\varPhi=\varPhi(x, t; k)$ is a matrix eigenfunction, $k$ is an iso-spectral parameter, the potential matrix $Q$ is
\begin{align}
Q=Q(x, t)=
\begin{bmatrix}
0&q(x, t) \\[0.05in]
\sigma q(-x,-t)&0
\end{bmatrix}, \quad \sigma=\pm 1,
\end{align}
and $\sigma_3$ is one of the Pauli matrices, which are
\begin{align} \no
\sigma_1=\begin{bmatrix}
0&1\\1&0
\end{bmatrix},\quad
\sigma_2=\begin{bmatrix}
0&-i\\i&0
\end{bmatrix},\quad
\sigma_3=\begin{bmatrix}
1&0\\
0&-1
\end{bmatrix}.
\end{align}
The zero curvature equation $X_t-T_x+[X, T]=0$ just leads to Eq.~(\ref{n-mKdV}). The only difference between the Lax pairs of nonlocal and local mKdV equations is that there are both the local function $q(x,t)$ and nonlocal function $q(-x, -t)$ in the nonlocal mKdV equation, which leads to their distinct wave structures and other properties.

Considering the asymptotic scattering problem ($x\to \pm\infty$) of Eqs.~(\ref{lax-x}) and (\ref{lax-t}):
\begin{align}\label{alax}
\begin{aligned}
\varPhi_x&=X_{\pm}(k)\varPhi, & X_{\pm}(k)&=\lim_{x\to\pm \infty}X(x, t; k)=ik\sigma_3+Q_{\pm}, \\[0.04in]
\varPhi_t&=T_{\pm}(k)\varPhi, &T_{\pm}(k)&=\lim_{x\to\pm \infty}T(x, t; k)=\left(4k^2+2\,\sigma\delta q_0^2\right)X_{\pm}(k),
\end{aligned}
\end{align}
with
$Q_{\pm}=\lim\limits_{x\to\pm \infty} Q(x, t)=\begin{bmatrix}0&q_{\pm}\\ \delta \sigma\, q_{\pm}&0 \end{bmatrix}$, 
we have the fundamental matrix solution of Eq.~(\ref{alax}) as
\begin{align}
\varPhi^{bg}(x, t; k)=
\left\{
\begin{alignedat}{2}
&E_{\pm}(k)\,\mathrm{e}^{i\theta(x, t; k)\sigma_3}, &&k\ne\sqrt{\delta\sigma}\,q_0,\\[0.04in]
&I+\left[x+\left(4k^2+2\,\sigma\delta q_0^2\right)t\right]X_{\pm}(k),&\quad&k=\sqrt{\delta\sigma}\,q_0,
\end{alignedat}\right.
\end{align}
where
\begin{align} \label{lp}
E_{\pm}(k)=
\begin{bmatrix}
1&\frac{iq_{\pm}}{k+\lambda} \vspace{0.05in} \\
-\delta\sigma\,\frac{iq_{\pm}}{k+\lambda}&1
\end{bmatrix},\quad \lambda(k)=\sqrt{k^2-\delta\sigma\,q_0^2},\quad
\theta(x, t; k)=\lambda\left[x+\left(4k^2+2\,\sigma\delta q_0^2\right)t\right].
\end{align}

For $\lambda(k)=\sqrt{k^2-\delta\sigma\,q_0^2}$, we clarify the two-sheeted Riemann surface in two cases corresponding to $\delta\sigma=1$ and $\delta\sigma=-1$.
Let $k-q_0=r_1\,\mathrm{e}^{i\theta_1}$ and $k+q_0=r_2\,\mathrm{e}^{i\theta_2}$. Then  $\lambda_m(k)=\sqrt{r_1r_2}\,\mathrm{e}^{i\left(\frac{\theta_1+\theta_2}{2}+m\pi\right)}$,  $m=0, 1$, respectively, are located on Sheets I and II.
\begin{itemize}
\item For $\delta\sigma=1$, the branch points are $k=\sqrt{\delta\sigma}\,q_0=\pm q_0$. Let $-\pi\le\theta_{1,2}<\pi$, then the branch cut (the discontinuity of $\lambda$) of the Reimann surface is the segment $\left[-q_0,q_0\right]$.

\item For $\delta\sigma=-1$, the branch points are $k=\sqrt{\delta\sigma}\,q_0=\pm iq_0$. Let $-\frac{1}{2}\,\pi\le\theta_{1,2}<\frac{3}{2}\,\pi$, then the branch cut (the discontinuity of $\lambda$)  of the Reimann surface is the segment $i\left[-q_0,q_0\right]$.
\end{itemize}

\begin{figure}[!t]
\centering
\includegraphics[scale=0.5]{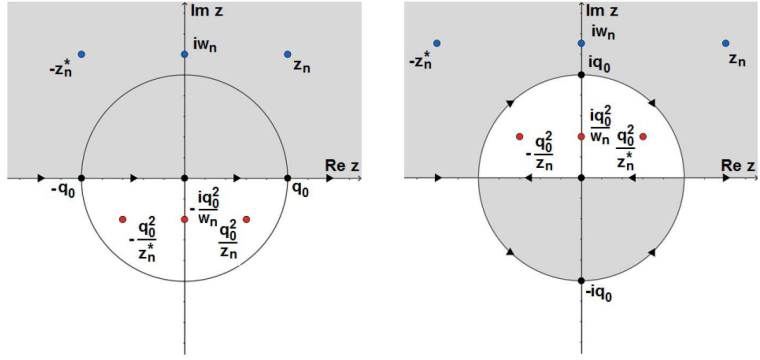}
\caption{The complex $z$-plane for $\delta\sigma=1$ (left) and $\delta\sigma=-1$ (right) showing the discrete spectrums [zeros of scattering data $s_{11}(z)$ (blue) in grey region and those of scattering data $s_{22}(z)$ (red) in white region], and the orientation of the contours for the Riemann-Hilbert problem. The grey ( $D_+$) and white ($D_-$ ) regions stand for $\mathrm{Im}\,\lambda(z)>0$ and $\mathrm{Im}\,\lambda(z)<0$, respectively.}
\label{z-plane}
\end{figure}

To seek for the analytical regions of the Jost solutions and scattering data, we  usually need to determine the regions where Im $\lambda(k)>0 \, (<0)$.  From the definition of the two-sheeted Riemann surface, one obtains that the region where Im $\lambda(k)>0$ is the upper-half plane (UHP) on Sheet-I and the lower-half plane (LHP) on Sheet-II  and the region where Im $\lambda(k)<0$ is the LHP on Sheet I and UHP on Sheet II.

The IST for the NLS equation with NZBCs was first presented by Zakharov in 1973~\cite{Zakharov1973}, where the two-sheeted Riemann surface was employed. After that a uniformization variable~\cite{Faddeev1987} was introduced to transform the scattering problem onto a standard complex $z$-plane. Let a uniformization variable $z$ be
\begin{align}\label{uvz}
z=k+\lambda=k+\sqrt{k^2-\delta\sigma\,q_0^2},
\end{align}
and the inverse mapping be deduced by
\begin{align}\label{kz}
k(z)=\frac{1}{2}\left(z+\delta\sigma\,\frac{q_0^2}{z}\right),\quad \lambda(z)=\frac{1}{2}\left(z-\delta\sigma\,\frac{q_0^2}{z}\right).
\end{align}
One can find the mapping relation between the two-sheeted Riemann $k$-surface and complex $z$-plane in two cases corresponding to $\delta\sigma=1$ and $\delta\sigma=-1$:
\begin{itemize}

\item As $\delta\sigma=1$, the mapping relation is observed as follows (see Fig. \ref{z-plane}(left)):

\begin{itemize}
\item The Sheet-I and Sheet-II are mapped onto the exterior and interior of the circle $|z|=q_0$, respectively;

\item The branch cut $\left[-q_0, q_0\right]$ is mapped onto the circle $|z|=q_0$;

\item $\left(-\infty, -q_0\right]\cup \left[q_0, +\infty\right)$ is mapped onto the real axis of $z$-plane;

\item The region where Im $\lambda(z)>0$ (Im $\lambda(z)<0$) is mapped onto the grey (white) domain in the $z$-plane.
\end{itemize}

\item As $\delta\sigma=-1$, the mapping relation is observed as follows (see Fig. \ref{z-plane}(right)):
\begin{itemize}

\item The Sheet-I and Sheet-II are mapped onto the exterior and interior of the circle $|z|=q_0$, respectively;

\item The branch cut $i\left[-q_0,q_0\right]$ is mapped into the circle $|z|=q_0$;

\item The real $k$ axis is mapped onto the real axis of $z$-plane;

\item The region for Im $\lambda(z)>0$ (Im $\lambda(z)<0$) is mapped onto the grey (white) domain in the $z$-plane.
\end{itemize}

\end{itemize}

For convenience, let
\begin{gather}
D_{\pm}=\left\{
\begin{aligned}
&\mathbb{C^{\pm}},&& &\mathrm{as}\quad \delta\sigma&=1,\\
\{z\in\mathbb{C}\,\big|\, & \pm (|z|-q_0)\,{\rm Im}\, z>0\},&& &\mathrm{as}\quad \delta\sigma&=-1.
\end{aligned}\right.
\end{gather}
which can generate $\mathrm{Im}\,\lambda>0$ ($\mathrm{Im}\, \lambda<0$) for $z\in D_+$  ($z\in D_-$).

\subsection{Jost solutions and modified forms}

Since the continuous spectrum of $X_{\pm}$ is the set of all values of $z$ satisfying $\lambda(z)\in\mathbb{R}$ \cite{Biondini2014}, then we denote the continuous spectrum by $\Sigma$. Then it follows from the above-mentioned analysis that we find
\begin{gather}
\Sigma=\left\{
\begin{aligned}
&\mathbb{R}\backslash\{0\},&& &\mathrm{as}\quad \delta\sigma&=1,\\
\left(\mathbb{R}\backslash\{0\}\right)&\cup\{z\in\mathbb{C}: |z|=q_0\},&&   &\mathrm{as}\quad \delta\sigma&=-1.
\end{aligned}\right.
\end{gather}
\begin{lemma}[Liouville's formula]\label{Liouville}
Let $y'=A(x)\,y$ be an $n$-dimensional first-order homogeneous linear differential equation
on an interval $\mathbb{D}$ of the real line, where $A(x)$ for $x\in \mathbb{D}$ denotes a square matrix of dimension $n$ with real or complex entries. If the trace $\mathrm{tr}\,A(x)$ is a continuous function, then a matrix-valued solution $\varPhi$ on $\mathbb{D}$ satisfies
$\mathrm{det}\, \varPhi(x)=\mathrm{det}\, \varPhi(x_0)\,\mathrm{e}^{\int_{x_0}^x\mathrm{tr}A(\xi)\,\mathrm{d}\xi},\quad x, x_0\in\mathbb{D}.$
\end{lemma}

\begin{proposition} The Jost solutions $\varPhi_{\pm}(x, t; z)$ satisfy simultaneously both parts of the Lax pair (\ref{lax-x},  \ref{lax-t}).
\end{proposition}

We know that the Jost solutions $\varPhi_{\pm}(x, t; z)$ satisfy
\begin{align}\label{Jost-asy}
\varPhi_{\pm}(x, t; z)=E_{\pm}(z)\,\mathrm{e}^{i\theta(x, t; z)\sigma_3}+o\left(1\right), \quad z\in\Sigma,\quad {\rm as}\quad x\to\pm\infty.
\end{align}
For  convenience, we introduce the modified Jost solutions $\mu_{\pm}(x, t; z)$ by eliminating the exponential oscillations:
\begin{align}\label{bianjie}
\mu_{\pm}(x, t; z)=\varPhi_{\pm}(x, t; z)\,\mathrm{e}^{-i\theta(x, t; z)\sigma_3},
\end{align}
such that
\begin{align}\label{mJost-asy}
\lim_{x\to\pm\infty}\mu_{\pm}(x, t; z)=E_{\pm}(z).
\end{align}
By the constant variation approach from Eq.~(\ref{lax-x}), one can write the modified Jost solutions as
\begin{align}\label{Jost-int}
\mu_{\pm}(x, t; z)=E_{\pm}(z)+\left\{
\begin{aligned}
&\int_{\pm\infty}^xE_{\pm}(z)\,\mathrm{e}^{i\lambda(z)(x-y)\widehat\sigma_3}\left[{E_{\pm}^{-1}(z)}\Delta Q_{\pm}(y, t)\,\mu_{\pm}(y, t; z)\right]\,\mathrm{d}y, && z\ne\sqrt{\delta\sigma}\,q_0,\\[0.05in]
&\int_{\pm\infty}^x\left[I+\left(x-y\right)X_{\pm}(z)\right]\Delta Q_{\pm}(y, t)\,\mu_{\pm}(y, t; z)\,\mathrm{d}y, &&  z=\sqrt{\delta\sigma}\,q_0,
\end{aligned}\right.
\end{align}
where $\Delta Q_{\pm}(x, t)=Q(x, t)-Q_{\pm}$, and $\mathrm{e}^{i\lambda(z)(x-y)\widehat\sigma_3}[{\bm\cdot}]=\mathrm{e}^{i\lambda(z)(x-y)\sigma_3}[{\bm\cdot}]\,\mathrm{e}^{-i\lambda(z)(x-y)\sigma_3}$. Next, we will establish the existence, uniqueness, continuity and analyticity of the Jost solutions. For convenience, we let $\Sigma^0=\Sigma\backslash\left\{\pm q_0\right\}$ as $\delta\sigma=1$ and $\Sigma^0=\Sigma\backslash\left\{\pm iq_0\right\}$ as $\delta\sigma=-1$.

\begin{lemma}\label{isjie}
Given $\sum_{n=0}^{\infty}A_n(x)$ and $B(x)$ on the interval $\mathbb{D}\subset\mathbb{R}$, where $A_n(x)$ and $B(x)$ are matrix-valued functions.  Suppose $\sum_{n=0}^{\infty}\left|\left|A_n(x)\right|\right|_1$ converges uniformly on the interval $\mathbb{D}$ and $\left|\left|B(x)\right|\right|_1$, $\left|\left|B(x)\,A_n(x)\right|\right|_1\in L^1\left(\mathbb{D}\right)$, then $\left|\left|B(x)\sum_{n=0}^{\infty}A_n(x)\right|\right|_1\in L^1\left(\mathbb{D}\right)$  and
$\int_\mathbb{D}\left[B(x)\sum_{n=0}^{\infty}A_n(x)\right]\mathrm{d}x=\sum_{n=0}^\infty\int_\mathbb{D}\left[B(x)\,A_n(x)\right]\,\mathrm{d}x.
$
\end{lemma}

\begin{proposition} \label{jiexi-m1}
Suppose $\left(\left|x\right|+1\right)\left(q(x, t)-q_{\pm}\right)\in L^1\left(\mathbb{R^{\pm}}\right)$, then the  Jost integral equation (\ref{Jost-int}) admits unique solutions $\mu_{\pm}(x, t; z)$ defined by (\ref{bianjie}) in $\Sigma$. Moreover, $\mu_{+1,-2}(x, t; z)$ can be extended analytically to $D_{+}$ and continuously to $D_{+}\cup\Sigma$, and $\mu_{-1, +2}(x, t; z)$ can be extended analytically to $D_{-}$ and continuously to $D_{-}\cup\Sigma$, where $\mu_{\pm j}(x, t; z),\, j=1,2$ is referred to the $j$th column of $\mu_{\pm}(x, t; z)$.
\end{proposition}

\begin{corollary}\label{jiexi-1}
Suppose $\left(\left|x\right|+1\right)\left(q(x, t)-q_{\pm}\right)\in L^1\left(\mathbb{R^{\pm}}\right)$, then  Eq.~(\ref{lax-x})
has unique solutions $\varPhi_{\pm}(x, t; z)$ given by Eq.~(\ref{Jost-asy}) in $\Sigma$. Besides, $\varPhi_{+1,-2}(x, t; z)$ can be extended analytically to $D_{+}$ and continuously to $D_{+}\cup\Sigma$ while $\varPhi_{-1, +2}(x, t; z)$ can be extended analytically to $D_{-}$ and continuously to $D_{-}\cup\Sigma$.
\end{corollary}

\subsection{The scattering matrix and reflection coefficients}

Since, from Lemma \ref{Liouville}, $\varPhi_{\pm}(x, t; z)$ are both solutions of the Lax pair (\ref{lax-x}, \ref{lax-t}) as $z\in\Sigma^0$, thus there exists a constant scattering matrix $S(z)=\left(s_{ij}(z)\right)_{2\times 2}$ (not depend on $x$ and $t$) between them such that
\begin{align}\label{Jostchuandi}
\varPhi_+(x, t; z)=\varPhi_-(x, t; z)\,S(z),
\end{align}
where $s_{ij}(z)$'s are called the scattering coefficients.

\begin{proposition}\label{jiexiS}
Suppose $q(x, t)-q_{\pm}\in L^1\left(\mathbb{R^{\pm}}\right)$. Then $s_{11}(z)$ ($s_{12}(z)$) can be extended analytically to $D_+$ ($D_-$) and continuously to $D_+\cup\Sigma^0$ ($D_-\cup\Sigma^0$). Moreover, both $s_{12}(z)$ and $s_{21}(z)$ are continuous in $\Sigma^0$.
\end{proposition}
\begin{proof}  It follows from Eq.~(\ref{Jostchuandi}) that we have the scattering coefficients in the Wronskian representations
\begin{align}\label{S-lie}
\begin{aligned}
s_{11}(z)=\frac{{\rm Wr}(\varPhi_{+1}(x, t; z), \varPhi_{-2}(x, t; z))}{1-\delta\sigma\,q_0^2/z^2},\quad
s_{12}(z)=\frac{{\rm Wr}(\varPhi_{+2}(x, t; z), \varPhi_{-2}(x, t; z))}{1-\delta\sigma\,q_0^2/z^2},\\[0.05in]
 s_{21}(z)=\frac{{\rm Wr}(\varPhi_{-1}(x, t; z), \varPhi_{+1}(x, t; z))}{1-\delta\sigma\,q_0^2/z^2},
\quad s_{22}(z)=\frac{{\rm Wr}(\varPhi_{-1}(x, t; z), \varPhi_{+2}(x, t; z))}{1-\delta\sigma\,q_0^2/z^2},
\end{aligned}
\end{align}
where ${\rm Wr}(\bm\cdot, \bm\cdot)$ denotes the Wronskian determinant. The proof follows trivally from Corollary \ref{jiexi-1}.
\end{proof}

\begin{corollary}
Suppose $\left(\left|x\right|+1\right)\left[q(x, t)-q_{\pm}\right]\in L^1\left(\mathbb{R^{\pm}}\right)$. Then $\lambda(z)\,s_{11}(z)$ ($\lambda(z)\,s_{22}(z)$) can be extended analytically to $D_+$ ($D_-$) and continuously to $D_+\cup\Sigma$ ($D_-\cup\Sigma$). Moreover, both $\lambda(z)\,s_{12}(z)$ and $\lambda(z)\,s_{21}(z)$ are continuous in $\Sigma$.
\end{corollary}

Since one can not exclude the possibilities of zeros for $s_{11}(z)$ and $s_{22}(z)$ along $\Sigma$. To solve the Riemann-Hilbert problem in the inverse process, we only consider the potentials without spectral singularities~\cite{Zhou1989}, i.e., $s_{11}(z)\ne 0,\, s_{22}(z)\ne 0$ for $z\in\Sigma$. The reflection coefficients are defined in $z\in\Sigma$ as
\begin{align}\label{fanshe}
\varrho(z)=\frac{s_{21}(z)}{s_{11}(z)}=\frac{{\rm Wr}(\varPhi_{-1}(x, t; z), \varPhi_{+1}(x, t; z))}{{\rm Wr}(\varPhi_{+1}(x, t; z), \varPhi_{-2}(x, t; z))},
 \qquad \tilde\varrho(z)=\frac{s_{12}(z)}{s_{22}(z)}=\frac{{\rm Wr}(\varPhi_{+2}(x, t; z), \varPhi_{-2}(x, t; z))}{{\rm Wr}(\varPhi_{-1}(x, t; z), \varPhi_{+2}(x, t; z))}.
\end{align}

\subsection{Symmetry reductions}

We here study the symmetry relations of the Jost solutions and scattering matrix for variables and isospectral parameter. The symmetries of the scattering matrix can be derived from ones of the Jost solutions. The symmetries of the Jost solutions are obtained by the reduction conditions of the Lax pair.
\begin{proposition}[Reduction conditions]\label{3RC}
The $X(x, t; z)$ and $T(x, t; z)$ in  Eqs.~(\ref{lax-x}) and (\ref{lax-t}) have  the three symmetry reductions on $z$-plane.
\begin{itemize}
\item The first symmetry reduction is
\begin{align}\label{first-RC}
X(x, t; z)=-\sigma_4\,X(-x, -t; -z^*)^*\,\sigma_4, \quad T(x, t; z)=-\sigma_4\,T(-x, -t; -z^*)^*\,\sigma_4,
\end{align}
where  $\sigma_4$ is defined as
$\sigma_4=\left\{
\begin{array}{l}
\sigma_1,\quad {\rm as} \,\,\,\, \sigma=-1,\\
\sigma_2,\quad {\rm as} \,\,\,\, \sigma=1.
\end{array}\right.$
\item The second symmetry reduction is
\begin{align}\label{second-RC}
X(x, t; z)=X(x, t; -z^*)^*,\quad T(x, t; z)=T(x, t; -z^*)^*.
\end{align}
\item The third symmetry reduction is
\begin{align}\label{third-RC}
X(x, t; z)=X\left(x, t; \delta\sigma\,\frac{q_0^2}{z}\right),\quad T(x, t; z)=T\left(x, t;  \delta\sigma\,\frac{q_0^2}{z}\right).
\end{align}
\end{itemize}
\end{proposition}

\begin{proposition}\label{jddcx}
The symmetries for Jost solutions $\varPhi_{\pm}(x, t; z)$ in $z\in\Sigma$ are found as follows:
\begin{itemize}
\item The first symmetry is
\begin{align}\label{Jduichen-1}
\varPhi_{\pm}(x, t; z)=\sigma_4\,\varPhi_{\mp}(-x, -t; -z^*)^*\,\sigma_4,
\end{align}
As $\sigma=-1$ (focusing case), the symmetries are
\begin{align}
\begin{aligned}
\varPhi_{\pm j}(x, t; z)=\sigma_1\,\varPhi_{\mp (3-j)}(-x, -t; -z^*)^*,\,\,\, j=1,2
\end{aligned}
\end{align}
As $\sigma=1$ (defocusing case), the symmetries are
\begin{align}
\begin{aligned}
\varPhi_{\pm j}(x, t; z)&=i(-1)^{j+1}\,\sigma_2\,\varPhi_{\mp (3-j)}(-x, -t; -z^*)^*,\,\,\, j=1,2
\end{aligned}
\end{align}
\item The second symmetry is
\begin{align}\label{Jduichen-2}
\varPhi_{\pm}(x, t; z)&=\varPhi_{\pm}(x, t; -z^*)^*.
\end{align}

\item The third symmetry is
\begin{align}\label{Jduichen-3}
\varPhi_{\pm}(x, t; z)=\frac{i}{z}\,\varPhi_{\pm}\left(x, t; \delta\sigma\,\frac{q_0^2}{z}\right)\sigma_3\,Q_{\pm},
\end{align}
column-wise, which reads
\begin{align}
\begin{gathered}
\varPhi_{\pm j}(x, t; z)=(-\delta\sigma)^j\,\frac{iq_{\pm}}{z}\,\varPhi_{\pm (3-j)}\left(x, t; \delta\sigma\,\frac{q_0^2}{z}\right),
\,\,\,j=1,2
\end{gathered}
\end{align}
\end{itemize}
\end{proposition}

\begin{proposition}\label{Sduichen}
The symmetries for the scattering matrix $S(z)$ in $z\in\Sigma$ are given as follows.
\begin{itemize}
\item  The first symmetry is
\begin{align}\label{Sduichen-1}
S(z)=\sigma_4\left[S\left(-z^*\right)^*\right]^{-1}\sigma_4,\,\, i.e., \,\, s_{ii}(z)=s_{i}(-z^*)^*,\quad s_{ij}(z)=\sigma\,s_{ji}(-z^*)^*,\,\,
i,j=1,2,\, i\not=j
\end{align}

\item The second symmetry is
\begin{align}\label{Sduichen-2}
S(z)=S(-z^*)^*, \quad i.e.,\quad
s_{ij}(z)=s_{ij}(-z^*)^*,\quad i, j=1, 2.
\end{align}
\item The third symmetry is
\begin{align}\label{Sduichen-3}
S(z)=\left(\sigma_3\,Q_-\right)^{-1}S\left(\delta\sigma\,\frac{q_0^2}{z}\right)\left(\sigma_3\,Q_+\right),
\end{align}
element-wise, which reads
\begin{align}
s_{11}(z)=\delta s_{22}\left(\delta\sigma\,\frac{q_0^2}{z}\right),\quad s_{12}(z)=-\sigma\,s_{21}\left(\delta\sigma\,\frac{q_0^2}{z}\right).
\end{align}
\end{itemize}
\end{proposition}

\subsection{Discrete spectrum with simple poles}

Suppose that $s_{11}(z)$ has $N_1$ and $N_2$ simple zeros, respectively, in $D_+\cap\left\{z\in\mathbb{C}:\mathrm{Re} \,z>0\right\}$ denoted by $z_n$, $n=1, 2, \cdots, N_1$ and in $D_+\cap\left\{z\in\mathbb{C}:\mathrm{Re} \,z=0\right\}$ denoted by $iw_n$, $n=1, 2, \cdots, N_2$. It follows from the symmetry relations of the scattering coefficients in Proposition \ref{Sduichen} that
\begin{align}
s_{11}(z_n)=s_{11}(-z_n^*)=s_{22}\left(\delta\sigma\,\frac{q_0^2}{z_n}\right)=s_{22}\left(-\delta\sigma\,\frac{q_0^2}{z_n^*}\right),\quad n=1,2,\cdots, N_1,
\end{align}
\begin{align}
s_{11}(iw_n)=s_{22}\left(-\delta\sigma\,\frac{iq_0^2}{w_n}\right),\quad n=1,2,\cdots, N_2.
\end{align}
Thus, the discrete spectrum is given by
\begin{align}\label{lisanpu}
Z=\left\{z_n,  \, -z_n^*,\, \delta\sigma\,\frac{q_0^2}{z_n},\, -\delta\sigma\,\frac{q_0^2}{z_n^*}\right\}_{n=1}^{N_1}\bigcup\left\{iw_n,\,-\delta\sigma\,\frac{iq_0^2}{w_n}\right\}_{n=1}^{N_2},
\end{align}
whose distribution is shown in Fig. \ref{z-plane}.

For  convenience, let
\begin{align}\label{bAdingyi}
b[z_0]=\left\{
\begin{aligned}
\frac{\varPhi_{+1}(x, t; z_0)}{\varPhi_{-2}(x, t; z_0)}, \quad z_0\in Z\cap D_+,\\[0.05in]
\frac{\varPhi_{+2}(x, t; z_0)}{\varPhi_{-1}(x, t; z_0)}, \quad z_0\in Z\cap D_-,
\end{aligned}\right.\qquad
A[z_0]=\left\{
\begin{aligned}
\frac{b[z_0]}{s_{11}'(z_0)},\quad z_0\in Z\cap D_+,\\[0.05in]
\frac{b[z_0]}{s_{22}'(z_0)},\quad z_0\in Z\cap D_-,
\end{aligned}\right.
\end{align}
where $\frac{\bm\cdot}{\bm\cdot}$ in the expression of $b[z_0]$ denotes the proportional coefficient. Then we can write the residue condition in the compact form:
\begin{align}\label{Jie-liushu}
\begin{aligned}
\mathop\mathrm{Res}\limits_{z=z_0}\left[\frac{\varPhi_{+1}(x, t; z)}{s_{11}(z)}\right]&=A[z_0]\,\varPhi_{-2}(x, t; z_0),\quad z_0\in Z\cap D_+,\\[0.05in]
\mathop\mathrm{Res}\limits_{z=z_0}\left[\frac{\varPhi_{+2}(x, t; z)}{s_{22}(z)}\right]&=A[z_0]\,\varPhi_{-1}(x, t; z_0),\quad z_0\in Z\cap D_-,
\end{aligned}
\end{align}

\begin{proposition}
For the given $z_0\in Z$, there exist three relations for $b[z_0]$, $s'_{11}(z_0)$ and $s'_{11}(z_0)$:
\begin{itemize}
\item The first relation is
\begin{align}
b\left[z_0\right]=-\frac{\sigma}{b\left[-z_0^*\right]^*}, \quad s_{11}'(z_0)=-s_{11}'\left(-z_0^*\right)^*, \quad s_{22}'(z_0)=-s_{22}'\left(-z_0^*\right)^*.
\end{align}
\item The second relation is
\begin{align}
b[z_0]=-b\left[-z_0^*\right]^*, \quad s_{11}'(z_0)=-s_{11}'\left(-z_0^*\right)^*, \quad s_{22}'(z_0)=-s_{22}'\left(-z_0^*\right)^*.
\end{align}
\item The third relation is
\begin{align}
b[z_0]=-\sigma\,b\left[\delta\sigma\,\frac{q_0^2}{z_0}\right], \quad s_{11}'(z_0)=-\sigma\,\frac{q_0^2}{z_0^2}\,s_{22}'\left(\delta\sigma\,\frac{q_0^2}{z_0}\right).
\end{align}
\end{itemize}
\end{proposition}

From the first relation, one concluds that imaginary discrete spectrum $iw_n$ exists if and only if $\sigma=-1$. That is to say that as $\sigma=1$, one has $N_2=0$. One has the following corollary.

\begin{corollary}\label{bguanxi}
The relations for $b[{\bm\cdot}]$ in $Z$ are given by
\begin{align}
\begin{gathered}
b\left[z_n\right]=b\left[-z_n^*\right]^*=-\sigma\,b\left[\delta\sigma\,\frac{q_0^2}{z_n}\right]=-\sigma\,b\left[-\delta\sigma\,\frac{q_0^2}{z_n^*}\right]^*, \quad b\left[z_n\right]^2=-\sigma,\vspace{0.1in}\\
s_{11}'(z_n)=-s_{11}'\left(-z_n^*\right)^*=-\sigma\,\frac{q_0^2}{z_n^2}\,s_{22}'\left(\delta\sigma\,\frac{q_0^2}{z_n}\right)
=\sigma\,\frac{q_0^2}{z_n^2}\,s_{22}'\left(-\delta\sigma\,\frac{q_0^2}{z_n^*}\right).
\end{gathered}
\end{align}
As $\sigma=-1$, one has
\bee
\begin{array}{l}
b\left[iw_n\right]=b\left[\delta\dfrac{iq_0^2}{w_n}\right], \quad b\left[iw_n\right]=1, -1,\,\,\,\,
s_{11}'\left(iw_n\right)=-\dfrac{q_0^2}{w_n^2}\,s_{22}'\left(\delta\dfrac{iq_0^2}{w_n}\right),\quad \mathrm{Re}\,s_{11}'(iw_n)=0
\end{array}
\ene
\end{corollary}

\subsection{Asymptotic behaviors}

 We here study the asymptotic behaviors of the modified Jost solutions and scattering matrix both $z\to 0$ and $z\to\infty$. Similar to Ref.~\cite{Demontis2014}, we consider the Neumann series:
\begin{align}
\mu_{\pm}(x, t; z)=\sum_{n=0}^\infty \mu_{\pm}^{[n]}(x, t; z)
\end{align}
with $\mu_{\pm}^{[0]}(x, t; z)=E_{\pm}(z)$ and
\begin{gather} \no
\mu_{\pm}^{[n+1]}(x, t; z)=\int_{\pm\infty}^xE_{\pm}(z)\,\mathrm{e}^{i\lambda(z)(x-y)\widehat\sigma_3}\left[E_{\pm}^{-1}(z)\Delta Q_{\pm}(y, t)\,\mu_{\pm}^{[n]}(y, t; z)\right]\mathrm{d}y,\,\,\, n=0,1,2,...
\end{gather}

Let $\mu_{\pm}^{[n], d}$ and $\mu_{\pm}^{[n], o}$ stand for diagonal and off-diagonal parts of $\mu_{\pm}^{[n]}$, respectively. One can find
\begin{align}
\begin{aligned}
&\mu_{\pm}^{[n+1], d}(x, t; z)\\
&=\frac{z^2}{z^2-\delta\sigma q_0^2}\Bigg[\int_{\pm\infty}^x\left(\Delta Q_{\pm}(y, t)\,\mu_{\pm}^{[n], o}(y, t; z)-\frac{i\sigma_3Q_{\pm}(y, t)}{z}\,\Delta Q_{\pm}(y, t)\,\mu_{\pm}^{[n], d}(y, t; z)\right)\,\mathrm{d}y\\
&\qquad +\frac{i\sigma_3}{z}\,Q_{\pm}\int_{\pm\infty}^x\mathrm{e}^{i\lambda(x-y)\widehat\sigma_3}\left(\Delta Q_{\pm}(y, t)\,\mu_{\pm}^{[n], d}(y, t; z)-\frac{i\sigma_3\,Q_{\pm}(y, t)}{z}\,\Delta Q_{\pm}(y, t)\,\mu_{\pm}^{[n], o}(y, t; z)\right)\,\mathrm{d}y\Bigg]\\[0.05in]
&=\left\{
\begin{aligned}
&O\left(\mu_{\pm}^{[n], o}(x, t; z)\right)+O\left(z^{-1}\mu_{\pm}^{[n], d}(x, t; z)\right)
+O\left(z^{-2}\mu_{\pm}^{[n], d}(x, t; z)\right)+O\left(z^{-3}\mu_{\pm}^{[n], o}(x, t; z)\right),&& z\to\infty,  \\[0.05in]
&O\left(z^2\mu_{\pm}^{[n], o}(x, t; z)\right)+O\left(z\mu_{\pm}^{[n], d}(x, t; z)\right)+O\left(z^2\mu_{\pm}^{[n], d}(x, t; z)\right)+O\left(z\mu_{\pm}^{[n], o}(x, t; z)\right),&& z\to 0,
\end{aligned}\right.
\end{aligned}
\end{align}
\begin{align}
\begin{aligned}
&\mu_{\pm}^{[n+1], o}(x, t; z)\\
=&\frac{z^2}{z^2-\delta\sigma q_0^2}\Bigg[\frac{i\sigma_3\,Q_{\pm}}{z}\int_{\pm\infty}^x\left(\Delta Q_{\pm}(y, t)\,\mu_{\pm}^{[n], o}(y, t; z)-\frac{i\sigma_3\,Q_{\pm}(y, t)}{z}\,\Delta Q_{\pm}(y, t)\,\mu_{\pm}^{[n], d}(y, t; z)\right)\,\mathrm{d}y\\[0.05in]
&\qquad+\int_{\pm\infty}^x\mathrm{e}^{i\lambda(x-y)\widehat\sigma_3}\left(\Delta Q_{\pm}(y, t)\,\mu_{\pm}^{[n], d}(y, t; z)-\frac{i\sigma_3\,Q_{\pm}(y, t)}{z}\Delta Q_{\pm}(y, t)\,\mu_{\pm}^{[n], o}(y, t; z)\right)\,\mathrm{d}y\Bigg]\\[0.05in]
&=\left\{
\begin{aligned}
&O\left(\frac{\mu_{\pm}^{[n], o}(x, t; z)}{z}\right)+O\left(\frac{\mu_{\pm}^{[n], d}(x, t; z)}{z^2}\right)+O\left(\frac{\mu_{\pm}^{[n], d}(x, t; z)}{z}\right)+O\left(\frac{\mu_{\pm}^{[n], o}(x, t; z)}{z^2}\right),&& z\to\infty, \\[0.05in]
&O\left(z\mu_{\pm}^{[n], o}(x, t; z)\right)+O\left(\mu_{\pm}^{[n], d}(x, t; z)\right)+O\left(z^3\mu_{\pm}^{[n], d}(x, t; z)\right)+O\left(z^2\mu_{\pm}^{[n], o}(x, t; z)\right),&&z\to 0,
\end{aligned}\right.
\end{aligned}
\end{align}

Combining with
\begin{align}
\begin{aligned}
\mu_{\pm}^{[0], d}(x, t; z)&=O\left(1\right),  \quad\mu_{\pm}^{[0], o}(x, t; z)=O\left(1/z\right), \quad z\to\infty, \\[0.05in]
\mu_{\pm}^{[0], d}(x, t; z)&=O\left(1\right),  \quad\mu_{\pm}^{[0], o}(x, t; z)=O\left(1/z\right), \quad z\to 0,
\end{aligned}
\end{align}
it follows that for $m\in\mathbb{N}$, one has
\bee
\begin{array}{lll}
\mu_{\pm}^{[2m], d}=O\left(z^{-m}\right), \,\,\, & \mu_{\pm}^{[2m], o}=O\left(z^{-(m+1)}\right),  &z\to\infty,  \\[0.1in]
\mu_{\pm}^{[2m+1], d}=O\left(z^{-(m+1)}\right),  \,\,\, & \mu_{\pm}^{[2m+1], o}=O\left(z^{-(m+1)}\right),&z\to\infty, \\[0.1in]
\mu_{\pm}^{[2m], d}=O\left(z^m\right), \,\,\, &\mu_{\pm}^{[2m], o}=O\left(z^{m-1}\right), & z\to 0 \\[0.1in]
\mu_{\pm}^{[2m+1], d}=O\left(z^{m}\right),\,\,\, & \mu_{\pm}^{[2m+1], o}=O\left(z^{m}\right), &z\to 0.
\end{array}
\ene
Then we find the asymptotics of the modified Jost solutions.

\begin{proposition}\label{Jzjianjin}
The asymptotic behaviors for the modified Jost solutions are obtained as
\bee\label{Jzjianjin-1}
\mu_{\pm}(x, t; z)=
\left\{\begin{array}{ll}
I+O\left(1/z\right), & z\to\infty, \vspace{0.06in}\\
\dfrac{i}{z}\,\sigma_3\,Q_{\pm}+O\left(1\right),& z\to 0.
\end{array}
\right.
\ene
\end{proposition}

\begin{corollary}\label{Szjianjin}
The asymptotic behaviors for the scattering data are given by
\bee \label{ien}
S(z)=
\left\{\begin{array}{ll}
I+O\left(1/z\right),  \quad & z\to\infty, \vspace{0.06in}\\
\delta I+O\left(z\right),\quad & z\to 0.
\end{array}\right.
\ene
\end{corollary}

\section{Inverse scattering problem with NZBCs}

\subsection{Matrix Riemann-Hilbert problem}

\begin{proposition}
By defining the sectionally meromorphic matrices
\begin{align}\label{RHP-M}
M(x, t; z)=\left\{
\begin{aligned}
M^+(x, t; z)=\left(\frac{\mu_{+1}(x, t; z)}{s_{11}(z)},\, \mu_{-2}(x, t; z)\right),\quad z\in D_+, \\[0.05in]
 M^-(x, t; z)=\left(\mu_{-1}(x, t; z),\, \frac{\mu_{+2}(x, t; z)}{s_{22}(z)}\right), \quad z\in D_-.
\end{aligned}\right.
\end{align}
the multiplicative matrix Riemann-Hilbert problem can be proposed as follows.
\begin{itemize}
\item Analyticity: $M(x, t; z)$ is analytic in $\left(D_+\cup D_-\right)\backslash Z$, and has simple poles in the discrete spectrum $Z$.

\item Jump relation:
\begin{align}\label{RHP-Jump}
M^-(x, t; z)=M^+(x, t; z)\left(I-J(x, t; z)\right), \quad z\in\Sigma,
\end{align}
with the jump matrix
$J(x, t; z)=\mathrm{e}^{i\theta(x, t; z)\widehat\sigma_3}
\begin{bmatrix}
0&-\tilde\varrho(z)\\[0.05in]
\varrho(z)&\varrho(z)\,\tilde\varrho(z)
\end{bmatrix}.
$
\item Asymptotics:
\begin{align}\label{RHP-Asy}
M^{\pm}(x, t; z)=
\left\{\begin{array}{ll}
I+O\left(\frac{1}{z}\right), \quad & z\to\infty, \vspace{0.05in}\\
\dfrac{i}{z}\,\sigma_3\,Q_-+O\left(1\right), \quad & z\to 0.
\end{array}\right.
\end{align}
\end{itemize}
\end{proposition}

To solve the Riemann-Hilbert problem, it is convenient to define $\widehat\eta_n=\delta\sigma\,\frac{q_0^2}{\eta_n}$ with
\begin{align}
\eta_n=\left\{
\begin{aligned}
&z_n, & n&=1, 2, \cdots, N_1,\\[0.05in]
-{}&z^*_{n-N_1}, &n&=N_1+1, N_1+2, \cdots, 2N_1,\\[0.05in]
&iw_{n-2N_1},& n&=2N_1+1, 2N_1+2, \cdots, 2N_1+N_2.
\end{aligned}\right.
\end{align}

\begin{theorem}
The solution of the Riemann-Hilbert problem (\ref{RHP-M}, \ref{RHP-Jump}, \ref{RHP-Asy}) can be written as
\begin{align}\label{RHP-jie}
\begin{aligned}
M(x, t; z)=M_0+\frac{1}{2\pi i}\int_\Sigma\frac{M^+(x, t; \zeta)\,J(x, t; \zeta)}{\zeta-z}\,\mathrm{d}\zeta,\quad z\in\mathbb{C}\backslash\Sigma,
\end{aligned}
\end{align}
where $\int_\Sigma$ stands for the integral along the oriented contour (see Fig.~\ref{z-plane}), and
\bee \no
M_0=I+\frac{i}{z}\,\sigma_3\,Q_-+\sum_{n=1}^{2N_1+N_2}\left[\frac{\mathop\mathrm{Res}\limits_{z=\eta_n}M^+(z)}{z-\eta_n}
+\frac{\mathop\mathrm{Res}\limits_{z=\widehat\eta_n}M^-(z)}{z-\widehat\eta_n}\right].
\ene
\end{theorem}

\begin{proof}
Eliminating the asymptotic behavior and pole contribution, the jump condition (\ref{RHP-Jump}) becomes
\begin{align}\label{jumpbianxing}
\begin{aligned}
M^-(z)-M_0=M^+(z)-M_0-M^+(z)J(z),
\end{aligned}
\end{align}

The left-hand side of Eq.~(\ref{jumpbianxing}) is analytic in $D_-$ and the first two terms in the right-hand side of Eq.~(\ref{jumpbianxing}) is analytic in $D_+$. Both of their asymptotics are $O\left(\frac{1}{z}\right)$ as $z\to\infty$ and $O(1)$ as $z\to 0$. It follows from Corollary \ref{Szjianjin} that $J(x, t; z)$ is $O\left(\frac{1}{z}\right)$  as $z\to\infty$, and $O(z)$ as $z\to 0$. Finally,  we can proof this Proposition by applying the Cauchy projectors $P_{\pm}[g](z)$ defined by
\begin{align}
P_{\pm}\left[g\right](z)=\frac{1}{2\pi i}\int_\Sigma\frac{g(\zeta)}{\zeta-(z\pm i0)}\,\mathrm{d}\zeta,
\end{align}
to Eq.~(\ref{jumpbianxing}) and the Plemelj's formulae, where the notation $z\pm i0$ denotes the limit chosen from the left/right of $z$.
\end{proof}

\subsection{Closing the system}

Eqs.~(\ref{bianjie}) and (\ref{Jie-liushu}) imply that the residues are given by
\begin{align}
\begin{aligned}
\mathop\mathrm{Res}\limits_{z=\eta_n}\left[\frac{\mu_{+1}(x, t; z)}{s_{11}(z)}\right]=A[\eta_n]\,\mu_{-2}(x, t; \eta_n)\,\mathrm{e}^{-2i\theta(x, t; \eta_n)},\,\,\,\,
\mathop\mathrm{Res}\limits_{z=\widehat\eta_n}\left[\frac{\mu_{+2}(x, t; z)}{s_{22}(z)}\right]=A[\widehat\eta_n]\,\mu_{-1}(x, t; \widehat\eta_n)\,\mathrm{e}^{2i\theta(x, t; \widehat\eta_n)},
\end{aligned}
\end{align}
from which the residue parts in the solution of the Riemann-Hilbert problem are written as
\begin{align}\label{liushuhe}
\frac{\mathop\mathrm{Res}\limits_{z=\eta_n}M^+(x, t; z)}{z-\eta_n}+\frac{\mathop\mathrm{Res}\limits_{z=\widehat\eta_n}M^-(x, t; z)}{z-\widehat\eta_n}=\left[C_n(z)\,\mu_{-2}(x, t; \eta_n), \, \widehat C_n(z)\,\mu_{-1}(x, t; \widehat\eta_n)\right]
\end{align}
with
\begin{align} \no
C_n(z)=\frac{A[\eta_n]\,\mathrm{e}^{-2i\theta(x, t; \eta_n)}}{z-\eta_n}, \qquad \widehat C_n(z)=\frac{A[\widehat\eta_n]\,\mathrm{e}^{2i\theta(x, t; \widehat\eta_n)}}{z-\widehat\eta_n}.
\end{align}

The second column in Eq.~(\ref{RHP-jie}) yields
\begin{align}\label{kuaile}
\mu_{-2}(x, t; z)=
\begin{bmatrix}
\dfrac{iq_-}{z} \v\\   1
\end{bmatrix}
+\sum_{n=1}^{2N_1+N_2}\widehat C_n(z)\,\mu_{-1}(x, t; \widehat\eta_n)+\frac{1}{2\pi i}\int_\Sigma\frac{\left(M^+J\right)_2(x, t; \zeta)}{\zeta-z}\,\mathrm{d}\zeta.
\end{align}

Moreover, it follows from Eq.~(\ref{Jduichen-3}) that we have
\begin{align}\label{kuaile-1}
\mu_{-2}(x, t; z)=\frac{iq_-}{z}\,\mu_{-1}\left(x, t; \delta\sigma\,\frac{q_0^2}{z}\right).
\end{align}

It follows from Eqs.~(\ref{kuaile}) and (\ref{kuaile-1}) that for $z=\eta_k, k=1, 2, \cdots, 2N_1+N_2$, one has
\begin{align}\label{yaoqiujiele}
\begin{bmatrix}
\dfrac{iq_-}{\eta_k} \v\\  1
\end{bmatrix}
+\sum_{n=1}^{2N_1+N_2}\left(\widehat C_n(\eta_k)-\frac{iq_-}{\eta_k}\,\delta_{k, n}\right)\mu_{-1}(x, t; \widehat\eta_n)+\frac{1}{2\pi i}\int_\Sigma\frac{\left(M^+J\right)_2(x, t; \zeta)}{\zeta-\eta_k}\,\mathrm{d}\zeta=0,
\end{align}
where $\delta_{k, n}$ is the Kronecker delta function. These equations for $k=1, 2, \cdots, 2N_1+N_2$ comprise a system of $2N_1+N_2$ equations with $2N_1+N_2$ unknowns $\mu_{-1}(x, t; \widehat\eta_n), n=1, 2, \cdots, 2N_1+N_2$, which together with Eqs.~(\ref{RHP-jie}) and (\ref{kuaile-1}), give a closed system of equations for $M(x, t; z)$ by means of the scattering data.

\subsection{Reconstruction formula}

We recover the potential in terms of the solution of the above-mentioned Riemann-Hilbert problem.

\begin{theorem}
The reconstruction formula for the nonlocal mKdV equation with NZBCs is given by
\begin{align}\label{recons}
q(x, t)=q_--i\sum_{n=1}^{2N_1+N_2}A[\widehat\eta_n]\,\mathrm{e}^{2i\theta(x, t; \widehat\eta_n)}\,\mu_{-11}(x, t; \widehat\eta_n)+\frac{1}{2\pi}\int_\Sigma\left(M^+J\right)_{12}(x, t; \zeta)\,\mathrm{d}\zeta.
\end{align}
\end{theorem}

\begin{proof}
Eqs.~(\ref{RHP-jie}, \ref{liushuhe}) derive the following asymptotic behavior of $M(x, t; z)$:
\begin{align}
M(x, t; z)=I+z^{-1}\,M^{(1)}(x, t; z)+O\left(1/z^2\right), \quad z\to\infty,
\end{align}
with
\begin{align}
\begin{aligned}
M^{(1)}(x, t; z)&=i\,\sigma_3\,Q_-+\frac{i}{2\pi}\int_\Sigma M^+(x, t; \zeta)\,J(x, t; \zeta)\,\mathrm{d}\zeta  \\[0.05in]
&\qquad +\sum_{n=1}^{2N_1+N_2}\left[A[\eta_n]\,\mathrm{e}^{-2i\theta(x, t; \eta_n)}\mu_{-2}(x, t; \eta_n), \,A[\widehat\eta_n]\,\mathrm{e}^{2i\theta(x, t; \widehat\eta_n)}\mu_{-1}(x, t; \widehat\eta_n)\right]
\end{aligned}
\end{align}
Substituting $M(x, t; z)\,\mathrm{e}^{i\theta(x, t; z)\sigma_3}$ into  Eq.~(\ref{lax-x}) yields
\begin{align}\label{haha}
M_x(x, t; z)+\frac{i}{2}\left(z-\delta\sigma\frac{q_0^2}{z}\right)M(x, t; z)\,\sigma_3=\left[\frac{i}{2}\left(z+\delta\sigma\frac{q_0^2}{z}\right)\sigma_3+Q\right]M(x, t; z).
\end{align}
Comparing with the coefficient of $z^0$, the proof follows.
\end{proof}

\subsection{The trace formulae and theta condition}

The trace formula is that the scattering coefficients $s_{11}(z)$ and $s_{22}(z)$ can be formulated in terms of the discrete spectrum $Z$ and scattering coefficients $s_{12}(z)$ and $s_{21}(z)$. For convenience, it is necessary to define
\begin{align}
\begin{aligned}
\beta^+(z)=s_{11}(z)\prod_{n=1}^{2N_1+N_2}\frac{z-\widehat\eta_n}{z-\eta_n}, \qquad
\beta^-(z)=s_{22}(z)\prod_{n=1}^{2N_1+N_2}\frac{z-\eta_n}{z-\widehat\eta_n}.
\end{aligned}
\end{align}

It follows that $\beta^+(z)$ and $\beta^-(z)$ are analytic and have no zeros, respectively, in $D_+$ and $D_-$. Eq.~(\ref{ien}) implies that the asymptotic behavior is $\beta^{\pm}(z)\to 1$ as $z\to\infty$. Taking the determinant for Eq. (\ref{Jostchuandi}) yields that $\beta^+(z)\,\beta^-(z)=1+s_{12}(z)\,s_{21}(z)$. And then taking its logarithms becomes
$\log\beta^+(z)+\log\beta^-(z)=\log\left[1+s_{12}(z)\,s_{21}(z)\right].$ Applying the Cauchy projectors and Plemelj's formulae, one has
\begin{align}
\log\beta^{\pm}(z)=\pm\frac{1}{2\pi i}\int_\Sigma\frac{\log\left[1-\tilde\varrho(\zeta)\varrho(\zeta)\right]}{\zeta-z}\,\mathrm{d}\zeta, \quad z\in D^{\pm}.
\end{align}
Hence, the trace formulae are given by
\begin{align}\label{trace-1}
\begin{aligned}
s_{11}(z)=e^{s(z)}\prod_{n=1}^{2N_1+N_2}\frac{z-\eta_n}{z-\widehat\eta_n},\qquad
s_{22}(z)=e^{-s(z)}\prod_{n=1}^{2N_1+N_2}\frac{z-\widehat\eta_n}{z-\eta_n}
\end{aligned}
\end{align}
with
\bee \no
s(z)=\frac{1}{2\pi i}\int_\Sigma\frac{\log\left[1+s_{12}(\zeta)\,s_{21}(\zeta)\right]}{\zeta-z}\,\mathrm{d}\zeta.
\ene

In what follows, we use the obtained trace formulae to find the asymptotic phase difference of boundary values $q_+$ and $q_-$ (also called `theta condition' in Ref.~\cite{Faddeev1987}). To this end, let $z\to 0$ in Eq.~(\ref{trace-1}). The left-hand side of Eq.~(\ref{ien}) yields $s_{11}(z)\to\delta$.
Note that
\begin{align} \no
\prod_{n=1}^{2N_1+N_2}\frac{z-\eta_n}{z-\widehat\eta_n}=
\prod_{n=1}^{N_1}\frac{\left(z-z_n\right)\left(z+z_n^*\right)}{\left(z-\delta\sigma\,\dfrac{q_0^2}{z_n}\right)
\left(z+\delta\sigma\,\dfrac{q_0^2}{z_n^*}\right)}\,\prod_{m=1}^{N_2}\frac{z-iw_m}{z+\delta\sigma\,\dfrac{iq_0^2}{w_m}}
\end{align}
which leads to
\begin{align} \no
\prod_{n=1}^{2N_1+N_2}\frac{z-\eta_n}{z-\widehat\eta_n}\to \left(\prod_{n=1}^{N_1}\frac{\left|z_n\right|^4}{q_0^4}\right)
 \left(\prod_{m=1}^{N_2}\left(-\delta\sigma\right)\,\frac{w_m^2}{q_0^2}\right) \quad \mathrm{as}\quad z\to 0.
\end{align}

The theta condition is
\begin{gather}\label{lsyueshu}
\left(\prod_{n=1}^{N_1}\frac{\left|z_n\right|^4}{q_0^4}\,\prod_{m=1}^{N_2}\left(-\delta\sigma\right)\,\frac{w_m^2}{q_0^2}\right)\exp\left(\frac{1}{2\pi i}\int_\Sigma\frac{\log\left[1+s_{12}(\zeta)\,s_{21}(\zeta)\right]}{\zeta}\,\mathrm{d}\zeta\right)=\delta.
\end{gather}
Eq. (\ref{lsyueshu}) is a constraint for the spectrum $Z$ and scattering coefficients $s_{12}(z), s_{21}(z)$.

In addition, since $s_{11}'(\eta)$ and $s_{22}'(\widehat\eta)$ are necessary for the residue conditions, one need to evaluate them.  Taking the logarithm and determinant for eq. (\ref{trace-1}), one yields that
\begin{gather}\label{sdaoshu}
\begin{gathered}
s_{11}'\left(\eta_j\right)=\frac{\prod_{m\ne j}\left(\eta_j-\eta_m\right)}{\prod_{m=1}^{2N_1+N_2}\left(\eta_j-\widehat\eta_m\right)}\,\exp\left(\frac{1}{2\pi i}\int_\Sigma\frac{\log\left[1+s_{12}(\zeta)\,s_{21}(\zeta)\right]}{\zeta-\eta_j}\,\mathrm{d}\zeta\right),\\[0.05in]
s_{22}'\left(\widehat\eta_j\right)=\frac{\prod_{m\ne j}\left(\widehat\eta_j-\widehat\eta_m\right)}{\prod_{m=1}^{2N_1+N_2}\left(\widehat\eta_j-\eta_m\right)}\,\exp\left(-\frac{1}{2\pi i}\int_\Sigma\frac{\log\left[1+s_{12}(\zeta)\,s_{21}(\zeta)\right]}{\zeta-\widehat\eta_j}\,\mathrm{d}\zeta\right).
\end{gathered}
\end{gather}

\section{Reflectionless potentials: solitons and breathers}

When the reflection coefficients $\varrho(z)$ and $\tilde\varrho(z)$ vanish identically, we can present  the explicit solutions. In this case, there is no jump (i.e., $J=0$) from $M^+(x, t; z)$ to $M^-(x, t; z)$ along  the continuous spectrum, and the inverse problem can be solved explicitly by using an algebraic system.

The case $\varrho(z)=\tilde\varrho(z)=0$ implies $J(x, t; z)=0$, in which Eq.~(\ref{yaoqiujiele}) reduces
\begin{align}\label{sle}
\sum_{n=1}^{2N_1+N_2}\left(\widehat C_n(\eta_k)-\frac{iq_-}{\eta_k}\,\delta_{k, n}\right)\mu_{-11}(x, t; \widehat\eta_n)=-\frac{iq_-}{\eta_k}, \quad k=1, 2, \cdots, 2N_1+N_2.
\end{align}

Let $H=\left(h_{kn}\right)_{(2N_1+N_2)\times(2N_1+N_2)}$, $\gamma=\left(\gamma_n\right)_{(2N_1+N_2)\times 1}$, $\beta=\left(\beta_k\right)_{(2N_1+N_2)\times 1}$, where
$h_{kn}=\widehat C_n(\eta_k)-\frac{iq_-}{\eta_k}\,\delta_{k, n}, \quad \gamma_n=\mu_{-11}(x, t; \widehat\eta_n), \quad \beta_k=-\frac{iq_-}{\eta_k}.$
Solving the system of linear equation (\ref{sle}), one obtains $\gamma=G^{-1}\beta$.

Let $\alpha=\left(\alpha_n\right)_{(2N_1+N_2)\times 1}$, where $\alpha_n=A[\widehat\eta_n]\,\mathrm{e}^{2i\theta(x, t; \widehat\eta_n)}$. Then from the reconstruction formula in Eq.~(\ref{recons}), we get the Theorem:

\begin{theorem}
The reflectionless potential of the nonlocal mKdV equation with NZBCs is deduced as
\begin{align}\label{solu1}
q(x, t)=q_-+\frac{\mathrm{det} \begin{bmatrix} H&\beta \vspace{0.05in}\\
\alpha^T&0 \end{bmatrix}} {\mathrm{det}\,H}\,i.
\end{align}
\end{theorem}

In the following, we exhibit the explicit solutions in four cases  corresponding to two different nonlinearities $\sigma$ (i.e., focusing and defocusing cases) and two distinct values $\delta$ of the phase differences:
\begin{itemize}

\item {} Case 1.\, $\sigma=-1$ and $\delta=-1$. In this case, the distribution of the discrete spectrum is shown in Fig. \ref{z-plane}(left).
Eq.~(\ref{lsyueshu}) yields the theta condition  as
\begin{gather}
\left(\prod_{n=1}^{N_1}\frac{\left|z_n\right|^4}{q_0^4}\right)\left(\prod_{m=1}^{N_2}\left(-1\right)\frac{w_m^2}{q_0^2}\right)=-1.
\end{gather}

\begin{itemize}

\item {} There exists $1$-eigenvalue solution if and only if $N_1=0, N_2=1$, in which one can obtain $w_1=q_0$. Then one has $\eta_1=iq_0$, $\widehat\eta_1=-iq_0$. Eq.  (\ref{sdaoshu}) implies that $s_{11}'(\eta_1)=\dfrac{1}{2iq_0}$ and $s_{22}'(\widehat\eta_1)=-\dfrac{1}{2iq_0}$. Corollary \ref{bguanxi} derives that $b[\eta_1]=1, -1$ and $b[\widehat\eta_1]=1, -1$. Let $q_-=q_0\,\mathrm{e}^{i\theta_-}$, $b[\eta_1]=\mathrm{e}^{i\theta_1}$, where $\theta_-, \theta_1\in\{0, \pi\}$. From (\ref{bAdingyi}), one obtains $A[\widehat\eta_1]=-2iq_0\,\mathrm{e}^{i\left(\theta_1\right)}$. Then $h_{11}=-\mathrm{e}^{2q_0\left(x+2q_0^2t\right)}\mathrm{e}^{i\left(\theta_1\right)}-\mathrm{e}^{i\theta_-}$, $\alpha_1=-2iq_0\mathrm{e}^{2q_0\left(x+2q_0^2t\right)}\mathrm{e}^{i\left(\theta_1+2\theta_-\right)}$ and $\beta_1=-\mathrm{e}^{i\theta_-}$. Substituting the above data into (\ref{solu1}), one obtains the $1$-eigenvalue solution of the focusing nonlocal mKdV equation (\ref{n-mKdV}) as
\begin{gather}
q(x, t)=-q_0\,\mathrm{e}^{i\theta_-}\tanh\left[q_0\left(x+2\,q_0^2\,t\right)+\frac{\theta_1+\theta_-}{2}\,i\right],
\end{gather}
which is singular only when $\theta_1+\theta_-=\pi$.  Fig.~\ref{case1}(a) illustrates a $1$-eigenvalue kink solution.

 \item {} As $N_1=N_2=1$, we obtain the $2$-eigenvalue solutions of the focusing nonlocal mKdV equation (\ref{n-mKdV}). In this case, we choose $z_1=q_0\,\mathrm{e}^{i\varphi_1}, \varphi_1\in\left(0, \frac{\pi}{2}\right)$, $w_1=q_0$, $b[z_1]=\mathrm{e}^{i\theta_1}, b[iw_1]=\mathrm{e}^{i\theta_2}, q_-=q_0\,\mathrm{e}^{i\theta_-}, \theta_1, \theta_2, \theta_- \in\left\{0, \pi\right\}$. Substituting them into Eqs.~(\ref{sdaoshu}) and (\ref{solu1}), then one can obtain the $2$-eigenvalue solution, which displays a interaction of a kink soliton and a soliton (see Figs.~\ref{case1}(b, c, d)). When $t<0$ (i.e., before the interaction, the wave profile consists of a kink soliton and a bright soliton, when $t=0$ (i.e., they have the strong interaction), the wave profile becomes the modified kink soliton. When $t>0$  (i.e., after the interaction), the wave profile  consists of a kink soliton and a grey soliton.

\end{itemize}

\begin{figure}[!t]
\centering
\includegraphics[scale=0.4]{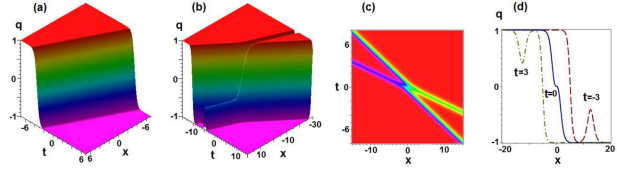}
\caption{(a) $1$-eigenvalue kink soliton solution for $q_0=1, \theta_1=\theta_-=0$. (b, c, d) The $2$-eigenvalue solution displaying  interaction of a kink soliton and a bright soliton for $q_0=1, \theta_-=0, \varphi_1=\frac{\pi}{4}, \theta_1=\pi, \theta_2=0$.}
\label{case1}
\end{figure}

\item {} Case 2.\, $\sigma=-1$ and $\delta=1$. Since $\delta\sigma=-1$, the discrete spectrum is shown in Fig. \ref{z-plane}(right). The theta condition (\ref{lsyueshu}) becomes
\begin{gather} \label{case2a}
\left(\prod_{n=1}^{N_1}\frac{\left|z_n\right|^4}{q_0^4}\right)\left(\prod_{m=1}^{N_2}\frac{w_m^2}{q_0^2}\right)=1,
\end{gather}

\begin{figure}[!t]
\centering
\vspace{0.1in}\includegraphics[scale=0.42]{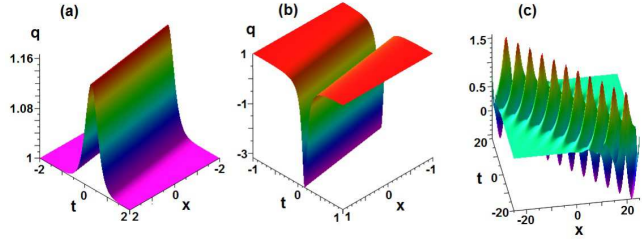}
\caption{(a) Bright soliton with $q_0=1, w_1=\frac{3}{2}, \theta_1=0, \theta_2=\pi, \theta_-=0$. (b) Dark soliton with $q_0=1, w_1=\frac{3}{2}, \theta_1=\pi, \theta_2=0, \theta_-=0$. (c) Breather solution with $q_0=\frac{1}{2},  \epsilon=\frac{1}{2}, \theta_1=0, \theta_2=\pi, \theta_-=0, \varphi_1=\frac{\pi}{6}.$}
\label{case2}
\end{figure}

\begin{itemize}

\item {} It follows from Eq.~(\ref{case2a}) that there is no $1$-eigenvalue solution.

\item {} For $N_1=0, N_2=2$, we construct a $2$-eigenvalue reflectionless potential. In this case, the theta condition yields $\left|w_1w_2\right|=q_0^2$. We take $w_1>q_0$ and then $w_2=-\frac{q_0^2}{w_1}$.  By the definition of $\eta_n$ and $\widehat\eta_n$, one can see that $\eta_1=iw_1, \eta_2=-\frac{iq_0^2}{w_1}, \widehat\eta_1=\frac{iq_0^2}{w_1}$ and $\widehat\eta_2=-iw_1$. Let $b[\eta_1]=\mathrm{e}^{i\theta_1}, b[\eta_2]=\mathrm{e}^{i\theta_2}, q_-=q_0\,\mathrm{e}^{i\theta_-}$, where $\theta_1, \theta_2, \theta_-\in\{0, \pi\}$. One obtains that $b[\widehat\eta_1]=\mathrm{e}^{i\left(\theta_1+2\theta_-\right)}$, $b[\widehat\eta_2]=\mathrm{e}^{i\left(\theta_2+2\theta_-\right)}$. Eq. (\ref{sdaoshu}) deduces that
\begin{gather}\no
s_{11}'(\eta_2)=\frac{i\left(q_0^2+w_1^2\right)}{2w_1\left(q_0^2-w_1^2\right)}, \quad s_{11}'(\eta_1)=\frac{iw_1\left(q_0^2+w_1^2\right)}{2q_0^2\left(q_0^2-w_1^2\right)}.
\end{gather}
By Corollary (\ref{bguanxi}), one can give
\begin{gather} \no
s_{22}'(\widehat\eta_1)=-\frac{iw_1\left(q_0^2+w_1^2\right)}{2q_0^2\left(q_0^2-w_1^2\right)}, \quad s_{22}'(\widehat\eta_1)=-\frac{i\left(q_0^2+w_1^2\right)}{2w_1\left(q_0^2-w_1^2\right)}.
\end{gather}

From Eq.~(\ref{bAdingyi}), one has
\begin{gather}\no
A\left[\widehat\eta_1\right]=\frac{2iq_0^2\left(q_0^2-w_1^2\right)}{w_1\left(q_0^2+w_1^2\right)}\,\mathrm{e}^{i\left(\theta_1+2\theta_-\right)}, \quad A\left[\widehat\eta_2\right]=\frac{2iw_1\left(q_0^2-w_1^2\right)}{\left(q_0^2+w_1^2\right)}\,\mathrm{e}^{i\left(\theta_2+2\theta_-\right)}.
\end{gather}

Substituting them into (\ref{solu1}), one can find the $2$-eigenvalue solution of the focusing nonlocal mKdV equation (\ref{n-mKdV})
\begin{gather}
q(x, t)=\frac{\mathrm{e}^{i\theta_-}}{w_1}\,\frac{w_1q_0\left(q_0^2+w_1^2\right)
 [\mathrm{e}^{2\varphi+i\left(\theta_1+\theta_2\right)}-\mathrm{e}^{2i\theta_-}]+2\left(w_1^4\mathrm{e}^{i\theta_2}
-q_0^4\mathrm{e}^{i\theta_1}\right)\mathrm{e}^{\varphi+i\theta_-}}   {\left(q_0^2+w_1^2\right)[\mathrm{e}^{2\varphi+i\left(\theta_1+\theta_2\right)}-\mathrm{e}^{2i\theta_-}]+2q_0w_1\left(\mathrm{e}^{i\theta_2}
-\mathrm{e}^{i\theta_1}\right)\mathrm{e}^{\varphi+i\theta_-}},
\end{gather}
where
\begin{gather} \no
\varphi=\frac{\left(w_1^2-q_0^2\right)\left[w_1^2x-\left(w_1^4+4w_1^2q_0^2+q_0^4\right)t\right]}{w_1^3},
\end{gather}
which are displayed in Figs.~\ref{case2}(a, b). Figs.~\ref{case2}(a) and (b) exhibit the bright and dark soliton structures, respectively,  and Fig. \ref{case2}(b) a dark structure.

\item {} For $N_1=2, N_2=0$, we obtain another $2$-eigenvalue solution of the focusing nonlocal mKdV equation (\ref{n-mKdV}). In this case, we take $z_1=\left(1+\epsilon\right)q_0\,\mathrm{e}^{i\varphi_1}, \epsilon>0, \varphi_1\in\left(0, \frac{\pi}{2}\right)$ and then it follows that $z_2=\frac{1}{1+\epsilon}\,q_0\,\mathrm{e}^{-i\varphi_1}$. Let $b[z_1]=\mathrm{e}^{i\theta_1}, b[z_2]=\mathrm{e}^{i\theta_2}, q_-=q_0\,\mathrm{e}^{i\theta_-}, \theta_1, \theta_2, \theta_- \in\left\{0, \pi\right\}$. Then this kind of $2$-eigenvalue solution  (\ref{solu1}) exhibits a breather structure (see Fig. \ref{case2}(c)).

\end{itemize}

\begin{figure}[!t]
\centering
\includegraphics[scale=0.4]{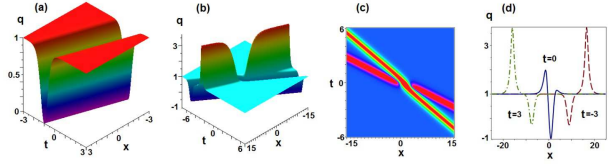}
\caption{(a) Dark soliton with $q_0=1,\, \theta_1=\pi,\, \varphi_1=\frac{\pi}{3}, \,\theta_-=0$. (b, c, d) Collision between $1$-dark soliton and $1$-bright soliton with $q_0=1,\, \theta_1=\theta_2=\pi,\, \varphi_1=\frac{\pi}{6},\, \varphi_2=\frac{\pi}{3},\,\theta_-=0$.}
\label{case3}
\end{figure}

\item Case 3.\, $\sigma=1$ and $\delta=1$. In this case, the distribution of discrete spectrum is shown in Fig. \ref{z-plane}(left) without imaginary discrete spectrum and the theta condition from (\ref{lsyueshu}) reads $\prod_{n=1}^{N_1}\frac{\left|z_n\right|^4}{q_0^4}=1.$

\begin{itemize}

\item {} When $N_1=1$, there exists $1$-eigenvalue solution of  the defocusing nonlocal mKdV equation (\ref{n-mKdV}). Let $z_1=q_0\,\mathrm{e}^{i\varphi_1}, \varphi_1\in\left(0, \frac{\pi}{2}\right)$. By the definition of $\eta_n$ and $\widehat\eta_n$, one has $\eta_1=q_0\,\mathrm{e}^{i\varphi_1}, \eta_2=-q_0\,\mathrm{e}^{-i\varphi_1}, \widehat\eta_1=q_0\,\mathrm{e}^{-i\varphi_1}$ and $\widehat\eta_2=-q_0\,\mathrm{e}^{i\varphi_1}$. Let $q_-=q_0\,\mathrm{e}^{i\theta_-}$, $b[\eta_1]=i\,\mathrm{e}^{i\theta_1}$, where $\theta_-, \theta_1\in\{0, \pi\}$. From Corollary \ref{bguanxi}, one obtains $b\left[\widehat\eta_1\right]=-i\,\mathrm{e}^{i\theta_1}$ and $b\left[\widehat\eta_2\right]=i\,\mathrm{e}^{i\theta_1}$. Eq. (\ref{sdaoshu}) and Corollary \ref{bguanxi} generate
    \begin{gather} \no
s_{11}'\left(\eta_1\right)=\frac{\cos(\varphi_1)}{2i\,q_0\sin(\varphi_1)\,\mathrm{e}^{i\varphi_1}}, \quad s_{22}'\left(\widehat\eta_1\right)=-\frac{\cos(\varphi_1)\,\mathrm{e}^{i\varphi_1}}{2i\,q_0\sin(\varphi_1)}, \quad s_{22}'\left(\widehat\eta_2\right)=-\frac{\cos(\varphi_1)\,\mathrm{e}^{-i\varphi_1}}{2i\,q_0\sin(\varphi_1)},
\end{gather}
The definition of $A\left[{\bm \cdot}\right]$ (\ref{bAdingyi}) reads
\begin{gather}\no
A\left[\widehat\eta_1\right]=-\frac{2q_0\,\sin(\varphi_1)}{\cos(\varphi_1)}\,\mathrm{e}^{i\left(\theta_1-\varphi_1\right)},\quad A\left[\widehat\eta_2\right]=\frac{2q_0\,\sin(\varphi_1)}{\cos(\varphi_1)}\,\mathrm{e}^{i\left(\theta_1+\varphi_1\right)}.
\end{gather}

Thus the $1$-eigenvalue solution of the defocusing nonlocal mKdV equation (\ref{n-mKdV}) can be found as
\begin{gather}
q(x, t)=q_0\,\mathrm{e}^{i\theta_-}+\frac{iq_0\sin(2\varphi_1)\left(\mathrm{e}^{2i\varphi_1}-1\right)^2\left(\mathrm{e}^{i\left(2\theta_1-\varphi_1\right)}
-\mathrm{e}^{i\varphi_1}\right)\mathrm{e}^{\phi+2i\theta_-}}
{i\phi_1\sin(2\varphi_1)\,\mathrm{e}^\phi-\phi_2\sin^2(\varphi_1)\,\mathrm{e}^{2\phi}+\phi_3\cos^2(\varphi_1)},
\end{gather}
where
\bee \no
\begin{array}{l}
\phi=\left[2\,q_0x+4\,q_0^3(1+2\cos^2(\varphi_1))t\right]\sin\varphi_1,\\[0.1in]
\phi_1=\mathrm{e}^{i\left(2\theta_1+\varphi_1+\theta_-\right)}-\mathrm{e}^{i\left(2\theta_1+3\varphi_1+\theta_-\right)}
-\mathrm{e}^{i\left(3\varphi_1+\theta_-\right)}
+\mathrm{e}^{i\left(\varphi_1+\theta_-\right)},\\[0.1in]
\phi_2=2\,\mathrm{e}^{i\left(\theta_1+2\varphi_1\right)}+\mathrm{e}^{i\left(\theta_1+4\varphi_1\right)}+\mathrm{e}^{i\theta_1},\\[0.1in]
\phi_3=\mathrm{e}^{i\left(\theta_1+2\theta_-\right)}-2\,\mathrm{e}^{i\left(\theta_1+2\varphi_1+2\theta_-\right)}
+\mathrm{e}^{i\left(\theta_1+4\varphi_1+2\theta_-\right)},
\end{array}
\ene
which displays a dark soliton (see Fig. \ref{case3}(a)).

\item {} When $N_1=2$, we  can obtain the $2$-eigenvalue solution. Let $z_1=q_0\,\mathrm{e}^{i\varphi_1}, z_2=q_0\,\mathrm{e}^{i\varphi_2}, \varphi_1, \varphi_2\in\left(0, \frac{\pi}{2}\right), b[z_1]=i\,\mathrm{e}^{i\theta_1}, b[z_2]=i\,\mathrm{e}^{i\theta_2}, q_-=q_0\,\mathrm{e}^{i\theta_-}, \theta_1, \theta_2, \theta_- \in\left\{0, \pi\right\}$. Substituting them into (\ref{solu1}) yields the $2$-eigenvalue solution, which exhibits the elastic collision of a dark soliton and a bright soliton (see Figs.~\ref{case3}(b, c, d)).
    When $t<0$ (i.e., before the interaction, the wave profile consists of a dark soliton and a bright soliton, and the dark soliton with smaller amplitude is located the left side of the bright soliton with larger amplitude,  when $t=0$ (i.e., they have the strong interaction), the wave profile consists of two bright solitons with smaller amplitudes and a dark solitons with larger amplitude, and the two bright solitons are located both sides of the dark solitons. When $t>0$  (i.e., after the interaction), the wave profile  consists of a bright soliton and a dark soliton, and the dark soliton with smaller amplitude is located the right side of the bright soliton with larger amplitude.

\end{itemize}

\item {} Case 4.\, $\sigma=1$ and $\delta=-1$.  In this case, the theta condition (\ref{lsyueshu}) reads
$\prod_{n=1}^{N_1}\frac{\left|z_n\right|^4}{q_0^4}=-1,$
which is a contradiction. Hence, there are no reflectionless potentials of the defocusing nonlocal mKdV equation (\ref{n-mKdV}).
\end{itemize}

\section{Conclusions and discussions}

In conclusion, we have proposed a systematical theory of the IST for the focusing and defocusing nonlocal mKdV equations with NZBCs at infinity.
 The scattering problem has been analyzed by means of a uniformization variable. The direct scattering is used to obtain the analytic properties and symmetries of the Jost solution and the scattering data, and the discrete spectrum. The inverse scattering problem can be solve in terms of a Riemann-Hilbert problem. The Cauchy projectors and Plemelj's formulae are used to derive the solution of the RHP. It follows from the residue condition that the imaginary discrete spectrum exist only when $\sigma=-1$. Finally, we study the reflectionless potentials in four cases, and show their dynamical wave structures in detail. For the fourth case $\sigma=1,\, \delta=-1$, there exist no reflectionless potentials.

In contrast to the respective treatment for the nonlocal NLS equation for four distinct cases due to Ablowitz, {\it et al.} \cite{Ablowitz2018}, we present the theory of IST for the focusing and defocusing nonlocal mKdV equations with NZBCs in a unified approach. Similarly, the IST of the complex nonlocal mKdV equation~\cite{Ablowitz2016} with NZBCs can be deduced as long as one ignores the second symmetry reduction in our theory. The unified theory for the nonlocal ISTs used in this paper can also be extended to other integrable nonlocal nonlinear wave equations with NZBCs.

\vspace{0.1in}
\baselineskip=15pt

\noindent {\bf Acknowledgements}

\vspace{0.05in}
This work was partially supported by the NSFC under grants Nos.11731014 and 11571346, and CAS Interdisciplinary Innovation Team.

\end{document}